%% file: main.tex
\documentclass{article}

\PassOptionsToPackage{numbers,compress}{natbib}
\usepackage[preprint]{neurips_2020}   %
\usepackage[utf8]{inputenc}   %
\usepackage[T1]{fontenc}      %
\usepackage{hyperref}         %
\usepackage{url}              %
\usepackage{xurl}             %
\usepackage{booktabs}         %
\usepackage{mathtools,amssymb}
\usepackage{amsfonts}         %
\usepackage{nicefrac}         %
\usepackage{microtype}        %
\usepackage{array,colortbl}
\usepackage{xcolor}
\usepackage{graphicx}         %
\usepackage{algorithm,algorithmicx,algpseudocode}
\usepackage[capitalise,noabbrev]{cleveref}
\usepackage{caption}
\usepackage{subcaption}
\usepackage{etoolbox}
\usepackage{enumitem}
\usepackage[section]{placeins}  %

\newcommand{\bb}{\mathbf{b}}

\usepackage{tikz}
\usetikzlibrary{arrows.meta,positioning,shapes.geometric,shapes.misc,fit,backgrounds,calc,decorations.pathreplacing}
\usepackage{pgfplots}
\pgfplotsset{compat=1.18}
\usepackage{siunitx}
\makeatletter
\renewcommand{\@notice}{}
\makeatother

\hypersetup{colorlinks=true, linkcolor=blue!60!black, citecolor=green!40!black, urlcolor=blue!60!black,
  pdftitle={Harbormaster: Evidence-Gated, Replay-Safe Maritime Anomaly Detection on AWS},
  pdfauthor={Arun Sharma}}
\pdftrailerid{}
\newtheorem{proposition}{Proposition}
\newenvironment{proof}{\par\noindent\textit{Proof.}\ }{\hfill$\square$\par\medskip}
\newcommand{\scope}[1]{\texttt{\scriptsize #1}}
\newcommand{\vmax}{v_{\max}}
\newcommand{\vreq}{v_{\mathrm{req}}}

\newcommand{\pphys}{p_{\mathrm{phys}}}
\newcommand{\Pdata}{P_{\mathrm{data}}}
\DeclareMathOperator{\clip}{clip}
\DeclareMathOperator{\logistic}{logistic}

\title{Harbormaster: Evidence-Gated, Replay-Safe Maritime Anomaly Detection on AWS}

\author{Arun Sharma\\University of Minnesota\\\texttt{arunshar@umn.edu}}

\begin{document}

\maketitle

\begin{abstract}
Ships broadcast their positions through the Automatic Identification System (AIS), and those reports can be false or missing. An operator who acts on an anomaly alert needs that alert to be attributable, reviewable, and recoverable after a failure. This paper describes Harbormaster, a production-shaped system on AWS that follows three rules. Physics checks run on every report before any learned model runs. The DynamoDB read store is an idempotent projection of PostgreSQL, and a guard on the log sequence number (LSN) of each change protects every write to it. A candidate model must pass a holdout gate and a shadow comparison before a canary gives it traffic, and a burn-rate check guards each canary step. The paper proves that replaying any prefix or suffix of the change log leaves each projected key at the value of its highest applied LSN. It also names the effects that this result does not cover, such as repeated cache invalidations and repeated audit rows. In a bounded AWS window, a one-hour soak returned 35,999 HTTP 200 responses to 36,000 requests, with a 95th-percentile (p95) client latency of 142.751~ms. In a separate bounded AWS run of 900~s, the stream path received a burst of 400~records/s, and its consumer lag later drained to zero. Every number in the paper carries a label that says where it was measured, and the paper lists the parts of the design that were never built.
\end{abstract}

\section{Introduction}\label{sec:intro}

Ships report their identity, position, speed, and course through the Automatic Identification System (AIS) \citep{riveiro2018maritime}. A vessel can switch its transponder off \citep{welch2022unseen}, and a report can also carry a position that no vessel could reach. Two historical cases show what these anomalies look like. In one track, two consecutive reports are 207.846~km and 815~s apart, which implies a speed of 255.026~m/s. In another track, one vessel stays silent for 14~days, 18~hours, and 59~minutes. Both cases come from a read-only case analysis, and Harbormaster did not score either of them.

An anomaly alert is useful only when a person can act on it safely. The alert must be attributable, so the operator can see which rule or model produced it and from which inputs. It must be reviewable, so an analyst can record a judgment without changing the automated score. It must also be recoverable, so the derived state stays correct after a crash, a redelivered message, or a replay. Harbormaster is built around one design question. How can a cloud pipeline turn untrusted AIS reports into alerts that stay attributable, reviewable, and recoverable when messages repeat and when models change?

The system makes three design choices. A fast physics check runs before the heavier model, and the deterministic scorer needs no learned model to return a result. PostgreSQL holds the authoritative registry, and its DynamoDB read store is an idempotent projection of the change log. A candidate model must pass gates that produce evidence before it receives traffic. The system is production-shaped, and its AWS parts ran only in attended, time-boxed windows. That environment is no longer running.

Every number in the paper carries a scope label. The label names a bounded AWS window, a local run, a synthetic drill on real project code, or the public repository. A reader can rerun only the public test suite, the local change data capture (CDC) tests, and the public drill scripts, and \cref{app:repro} states that boundary. The paper makes five contributions.

\begin{itemize}[leftmargin=*,itemsep=2pt,topsep=2pt]
\item It describes a physics-first scoring path in which a required-speed check and a stream-layer physics gate run before any learned model, and the gap detector uses space-time prisms to merge candidate gaps (\cref{sec:scoring}).
\item It gives a replay-safe CDC path. Its DynamoDB projection accepts a change only when that change's LSN is newer, and the paper proves that replay leaves each key at its highest applied LSN under stated assumptions (\cref{sec:cdc}).
\item It specifies an evidence-gated promotion ladder of holdout, reward-hacking, shadow, and canary checks, and it tests the holdout, shadow, and canary steps in synthetic drills (\cref{sec:promotion}).
\item It documents a scale-to-zero design for an asynchronous model endpoint, together with two metric-dimension bugs that each disabled one scaling path without an error (\cref{sec:scale}).
\item It reports a scoped evaluation in which every result names where it ran, and it lists the parts of the design that were never built (\cref{sec:eval,sec:lessons}).
\end{itemize}

\section{Related Work}\label{sec:related}

\paragraph{Maritime anomaly detection.}
\citet{riveiro2018maritime} review maritime anomaly detection work on data, methods, systems, and user aspects. TREAD \citep{pallotta2013tread} learns shipping routes from AIS and flags vessels that leave them. \citet{welch2022unseen} map where fishing vessels disable their AIS transponders, which shows that a reporting gap can itself be the behavior of interest. Harbormaster's corridor detector follows the route-learning idea, and its lane builder uses the Douglas-Peucker algorithm \citep{douglas1973rdp} and HDBSCAN \citep{campello2013hdbscan}.

\paragraph{Space-time prisms and trajectory gaps.}
Time geography describes the set of places that a moving object can reach between two observed points \citep{hagerstrand1970people}, and \citet{miller1991prism} brought the space-time prism into geographic information systems. \citet{sharma2022rendezvous} used prisms to bound the region where two vessels could have met during a trajectory gap. \citet{sharma2024physics} detected abnormal trajectory gaps with physics-based bounds on vessel motion, through the STAGD gap detector and a dynamic region merge. These two trajectory-gap studies share an author with this paper. Harbormaster uses the prism test and the gap detection of that work inside a serving system, and it proposes no new detection method.

\citet{sharma2022tighter} tighten possible-rendezvous regions on spatial networks through time slicing and bidirectional pruning. Their method bounds where objects could have met during reporting gaps, while Harbormaster's serving path scores incoming vessel reports.

\paragraph{Learned spatial models and physical knowledge.}
\citet{sharma2025pidpm} study trajectory anomaly detection with physics-informed diffusion and kinematic constraints. \citet{yang2025gcdm} integrate road-map attributes into conditional diffusion for vehicle-trajectory generation, including geometric and dynamic properties. \citet{ghosh2024kicdpm} combine Kriging-based conditioning and variogram-based regularization for regional sea-level downscaling. These studies connect learned spatial models to domain knowledge across distinct prediction and generation tasks. Harbormaster addresses the serving problem of applying evidence gates and replay-safe state updates around model artifacts. Its reported AWS runs use a numpy surrogate for the data term, and its diffusion training plane remains design only.

\paragraph{Cloud geospatial analysis and parallel mining.}
\citet{sharma2018webglobe} present WebGlobe, a browser-based climate analysis system that moves computation to cloud servers. \citet{sharma2022stdmsurvey} survey spatiotemporal pattern families, including outliers, and their sequential and parallel formulations. These works connect spatial analysis to the infrastructure needed to process large datasets and expose results to users. Harbormaster studies operational correctness for a narrower maritime serving path through replay handling, artifact gates, and review records.

\paragraph{Stream processing and change data capture.}
Apache Flink \citep{carbone2015flink} runs stateful stream programs as pipelined dataflows, and its snapshots keep keyed state consistent after a failure \citep{carbone2017state}. MillWheel \citep{akidau2013millwheel} and the Dataflow model \citep{akidau2015dataflow} define event time and watermarks for late data. Kafka \citep{kreps2011kafka} carries the change events that Debezium \citep{debezium} reads through PostgreSQL logical replication \citep{postgreslogical}. DBLog \citep{andreakis2020dblog} uses watermarks to interleave table snapshots with log events. Harbormaster handles the overlap between Debezium's snapshot rows and streamed changes with a simpler rule, which applies every snapshot row at LSN zero. \citet{kleppmann2017ddia} treats derived data as a view of an ordered log. This paper states that view as a proposition for one sink built on DynamoDB conditional writes \citep{elhemali2022dynamodb}.

\paragraph{Machine learning operations and gating.}
\citet{sculley2015hidden} describe the maintenance debt that builds up around machine learning code. The ML Test Score \citep{breck2017mltest} turns deployment readiness into a rubric of tests, and TFX \citep{baylor2017tfx} places data and model validation inside the training pipeline. Engineers interviewed by \citet{shankar2022operationalizing} describe a shadow stage, in which a new model makes live predictions that users do not see. That stage comes before a rollout to a small share of live users. Canary analysis \citep{davidovic2018canary} compares a new version with the old one before a full rollout. Gandalf \citep{li2020gandalf} correlates fault signals with ongoing rollouts and stops a rollout that it judges unsafe. Harbormaster's canary uses the multiwindow burn-rate alerts of \citet{thurgood2018alerting}, which build on the service level objectives of site reliability engineering \citep{beyer2016sre}. Its holdout gate uses AUC \citep{fawcett2006roc} and the continuous ranked probability score (CRPS) \citep{gneiting2005crps,gneiting2007scoring}, and its endpoint follows the serverless goal of holding no idle capacity \citep{jonas2019serverless}.

\paragraph{Drift and reward hacking.}
Dataset shift \citep{quinonero2009shift} and concept drift \citep{gama2014drift} explain why a model that passed its holdout can degrade after deployment. \citet{rabanser2019failing} compare shift detectors and find that two-sample tests work best on the outputs of a pretrained classifier. Harbormaster's input check pairs the Population Stability Index \citep{yurdakul2020psi} with the Kolmogorov-Smirnov statistic \citep{massey1951ks}. \citet{amodei2016concrete} name reward hacking as a concrete safety problem, and \citet{skalse2022reward} give it a formal definition. Harbormaster's probe checks one narrow form of it, in which a preference-tuned candidate raises its reward while it respects physics less.

\section{System Overview}\label{sec:overview}

\input{figures/fig_architecture}

Harbormaster splits its work into two planes (\cref{fig:architecture}). The training plane belongs on an off-cloud research GPU cluster, and this paper treats it as design only. The serving plane runs on AWS, and its parts ran in bounded windows. No trained checkpoint was promoted into AWS serving, so the model endpoint served a labeled demo stand-in image.

\paragraph{Five workflows.}
The design has five workflows, and each one has its own trigger. Workflow~1 checks each incoming vessel report, and Workflow~2 saves an analyst's review label. Workflow~3 carries committed registry changes into the serving copies, and Workflow~4 builds shipping-lane context from historical batches. In the design, Workflow~5 submits a separate model job to an asynchronous endpoint, and its job service, dispatch worker, and result collector are not built. The ingestor replays only a synthetic replay fixture of 2,709 events for 8 vessels over 6~hours, and no real-time AIS source was connected.

\paragraph{Two brokers and the stream job.}
Position reports and database changes travel on separate brokers. Provisioned Amazon Kinesis \citep{awskinesis} carries position reports, and Managed Flink 1.20 \citep{awsmanagedflink} consumes them with state keyed by the vessel's MMSI (Maritime Mobile Service Identity). Kafka carries Debezium change events. The AWS windows used Amazon MSK Serverless, and the local stack runs a single-node KRaft Kafka that Strimzi manages on kind. The Flink job keeps the last five fixes for each vessel, and it updates that history before it applies the physics gate. For each report, one keyed operator computes features and applies the gate. For each report that passes, it writes a feature item to DynamoDB and posts a signed request to the scorer.

\paragraph{The scoring front door.}
The scorer is the entry point for every scoring request, and the project calls it the front door. It is a FastAPI service on ECS Fargate, and the Terraform at the reviewed commit runs it on Fargate Spot capacity. Requests reach it through an API Gateway HTTP API with IAM authorization and a VPC Link that resolves the service through Cloud Map, so the path has no standing load balancer. The deterministic scorer can call an asynchronous SageMaker endpoint for the data term $\Pdata$. The ECS task on AWS set no endpoint, so $\Pdata$ came from a numpy surrogate there. In the design, review cases go to a PostgreSQL table, and the Streamlit console that reads them ran only locally. On AWS, the review queue fell back to process memory, and \cref{sec:lessons} explains why.

\paragraph{Batch lake and control.}
In the design, Workflow~4 runs PySpark on EMR Serverless and writes lane tables to Apache Iceberg \citep{apacheiceberg}. EMR Serverless stops itself after 15 idle minutes. In a bounded AWS window, a valid backfill job succeeded, and a job with invalid data failed before any write. The lane builder uses a Douglas-Peucker tolerance of 200~m and an HDBSCAN minimum cluster size of 2. The publisher from lane tables to a serving reference is assumed integration, so the scorer reads a small synthetic lane reference. Terraform manages the infrastructure as code \citep{rahman2019iac}, and every phase sits behind a toggle that an apply turns off for teardown.

\section{Physics-First Scoring}\label{sec:scoring}

The scorer uses physics as its first source of evidence. Each rule in this section is a closed-form function of the reported fixes and the reference data, so an analyst can recompute any score from the constants in the public code.

\paragraph{Required speed and the stream gate.}
A fix $i$ has latitude $\phi_i$, longitude $\lambda_i$, and time $t_i$. For consecutive fixes $A$ and $B$, the system computes the haversine distance $d_{AB}$ on a sphere of radius $R$ and the speed $\vreq$ that the move requires.
\begin{align}
d_{AB} &= 2R\arcsin\sqrt{\sin^2\!\tfrac{\phi_B-\phi_A}{2}+\cos\phi_A\cos\phi_B\,\sin^2\!\tfrac{\lambda_B-\lambda_A}{2}}, \label{eq:haversine}\\
\vreq &= \frac{d_{AB}}{\max(t_B-t_A,\,\epsilon)}, \qquad \pphys = \min\!\left(1,\ \frac{\vmax}{\max(\vreq,\,\epsilon)}\right). \label{eq:pphys}
\end{align}
The speed cap is $\vmax=25$~knots, which is 12.861~m/s. The constant $\epsilon=10^{-6}$ is in seconds inside the first maximum and in m/s inside the second. The plausibility $\pphys$ equals 1 when a vessel at the cap could make the move, and it falls as the required speed rises. The Flink job forwards a report to the scorer only when $\pphys\ge 0.3$.

This gate delays the scoring of fast jumps, and it can also drop them. It drops every report whose required speed from the previous stored fix exceeds $\vmax/0.3$, which is about 83~knots. The scorer therefore does not see that report when it arrives. The Flink job still adds the dropped fix to the vessel's five-fix history. The next report that passes the gate carries that history to the scorer, and the speed detector uses the fastest segment in it. A jump between about 83 and 300~knots is therefore scored late, when a later report passes the gate. The jump goes unscored when the vessel's next four reports also fail the gate or never arrive, because the history then no longer holds both fixes of the jump. A jump above $12\,\vmax$, which is 300~knots, makes every later request that still carries both of its fixes fail with HTTP 422.

\input{figures/fig_prism}

\paragraph{Space-time prism feasibility.}
The space-time prism \citep{hagerstrand1970people,miller1991prism} is the set of places that a vessel could occupy between two fixes if it never exceeded $\vmax$. Its footprint on the map is an ellipse with foci at $A$ and $B$, which \cref{fig:prism} draws. The reach length $L$ sets the size of that ellipse.
\begin{align}
L &= \vmax\,(t_B-t_A), \qquad \text{feasible} \iff d_{AB}-\delta\le L, \label{eq:prism}\\
a &= \tfrac{1}{2}L, \qquad b=\tfrac{1}{2}\sqrt{\max\!\big(0,\,L^2-d_{AB}^2\big)}. \label{eq:axes}
\end{align}
The terms $a$ and $b$ are the semi-major and semi-minor axes when the pair is feasible, and the code allows a tolerance of $\delta=10^{-6}$~m. An infeasible pair means that no vessel under the cap could have made the move. On the single-vessel scoring path, the scorer computes this prism for the latest pair of fixes, and the fused score does not read the result. The feasibility test for that pair matches the speed test $\vreq\le\vmax$ up to the tolerance $\delta$, so the speed detector already scores the same condition. The prism geometry does matter inside the gap detector, which merges candidate gaps whose prism bounding rectangles overlap.

\paragraph{Speed severity.}
The speed detector uses the fastest segment in the supplied history. Let $v^\star$ be the largest $\vreq$ over consecutive fixes.
\begin{align}
s_{\mathrm{speed}} &= \begin{cases} 0, & v^\star\le\vmax,\\[2pt] \min\!\big(1,\ (v^\star/\vmax-1)/3\big), & v^\star>\vmax.\end{cases} \label{eq:speed}
\end{align}
The severity saturates at $4\,\vmax$. The scorer rejects a request with HTTP 422 when any segment exceeds $12\,\vmax$, and that request receives no score.

\paragraph{Gap severity.}
A silence longer than the 600~s coverage threshold opens a candidate gap. The gap detector follows the abnormal gap detection of \citet{sharma2024physics}, whose abnormal gap measure bounds a vessel's possible movement during a gap with a space-time prism. Harbormaster keeps the name abnormal gap measure (AGM) for its own formula, which blends kinematic implausibility with a data term.
\begin{align}
\mathrm{AGM} &= w\,(1-\pphys)+(1-w)\,\Pdata, \qquad w=0.6, \label{eq:agm}\\
\Pdata &= \logistic\!\left(\frac{\ln(1+d)+0.001\,\Delta t-12}{3}\right), \label{eq:pdata}\\
s_{\mathrm{gap}} &= \clip\!\big(0.5\,f_{\mathrm{dur}}+0.5\,\mathrm{AGM}\big), \qquad f_{\mathrm{dur}}=\clip\!\left(\frac{\Delta t-600}{7200-600}\right). \label{eq:gap}
\end{align}
Here $d$ is the gap distance in meters, $\Delta t$ is the gap duration in seconds, and $\clip$ limits a value to $[0,1]$. The term $\pphys$ applies the plausibility rule of \cref{eq:pphys} to the two fixes that bound the gap. The detector merges candidate gaps whose prisms overlap and scores the earliest gap in each merged group. The scorer then keeps the group with the largest AGM. \Cref{eq:pdata} is a numpy surrogate. The target design replaces it with a diffusion model trained on the off-cloud cluster, and on AWS the $\Pdata$ term came from the surrogate.

\paragraph{Corridor severity and fusion.}
The corridor detector measures the distance $d_c$ in meters from the current fix to the nearest lane segment. It can also add an unexpected-turn reason when the history is long enough and the vessel turns sharply far from any waypoint. A watchlist lookup also runs on every request outside the plan, and a watchlist hit or a sanctions hit each adds a reason with a fixed severity. The scorer fuses the severities $s_i$ of all reasons with a noisy-OR rule.
\begin{align}
s_{\mathrm{corr}} &= \clip\!\left(\frac{d_c-2000}{6000-2000}\right), \label{eq:corridor}\\
S &= 1-\prod_{i}(1-s_i), \qquad c = 1-2\min(S,\,1-S), \label{eq:fusion}\\
\text{review} &\iff c<0.7\ \lor\ S\ge 0.6. \label{eq:review}
\end{align}
The noisy-OR treats the reasons as separate pieces of evidence, so one strong reason can raise the score close to 1 by itself. The confidence $c$ is 1 at the ends of the score range and 0 at $S=0.5$. The confidence $c$ depends only on $S$, so the clause $c<0.7$ holds exactly when $0.15<S<0.85$. \Cref{eq:review} therefore sends a case to review exactly when $S>0.15$, and the clause $S\ge 0.6$ changes the outcome only when $S\ge 0.85$. The analyst's label is stored next to the case, and it never overwrites the automated score.

\paragraph{The deterministic planner.}
The planner chooses detectors from the length of the history, and it uses no language model. It always runs the corridor detector. It adds the prism, speed, and validate steps when the history holds at least one prior fix. It adds the gap detector when the history holds at least three. \Cref{alg:planner} in \cref{app:method} lists the full rule. The validate step runs a kinematic region gate that rejects anchors that imply more than $1.05\,\vmax$. The single-vessel scoring path passes it no anchors, so the step has no effect on that path.

\section{Replay-Safe Change Data Capture}\label{sec:cdc}

\input{figures/fig_cdc}

PostgreSQL 16 on RDS holds the authoritative vessel, watchlist, and sanctions tables, and every other copy of that data is derived (\cref{fig:cdc}). Debezium runs inside Kafka Connect and reads committed changes through logical replication \citep{postgreslogical,debezium}. Kafka delivers each change event at least once, so a consumer restart or a replay can deliver the same event again. In the code, a registry worker takes three actions for each event. It attempts a conditional write of a whole item to the DynamoDB read store, removes the stale Redis cache entry, and appends a row to an Iceberg audit table.

\paragraph{The commit protocol.}
The worker applies every event in a batch, flushes every sink, and only then commits its Kafka offsets. When a sink fails, the offsets stay uncommitted, and Kafka redelivers the batch. The worker skips Kafka tombstones, which are compaction metadata. A source delete becomes a soft-delete item that the worker writes at the delete's own LSN.

\paragraph{The LSN guard.}
Each item in the read store keeps the LSN of the change that last wrote it. The worker writes with a DynamoDB conditional put \citep{elhemali2022dynamodb}. The condition compares the stored LSN $\ell^\star(k)$ of key $k$ with the incoming LSN $\ell$.
\begin{align}
\textsc{apply}(k,\ell,v) &\iff \ell^\star(k)=\bot\ \lor\ \ell^\star(k)<\ell. \label{eq:guard}
\end{align}
Here $\bot$ means that the item has no stored LSN yet. A rejected write signals a stale or duplicate delivery, so the worker records the event as not applied and raises no error. The worker applies rows from Debezium's initial snapshot at $\ell=0$, and the audit row keeps the snapshot's own LSN. A snapshot row therefore applies only to an empty key, and every streamed change outranks it.

\begin{proposition}[Replay convergence]\label{prop:replay}
Consider a change log of events $e_j=(k_j,\ell_j,v_j)$ and a projection that starts empty. Assume that (i) for each key, an LSN identifies at most one value, (ii) the check and the write in \cref{eq:guard} happen atomically for each item, (iii) each accepted write replaces the whole item, and (iv) every delivered event maps to a valid item key. Let $\mathcal{D}$ be any finite sequence of deliveries drawn from the log, in any order and with any repetition. After the worker processes $\mathcal{D}$, each key $k$ that appears in $\mathcal{D}$ holds the value of the event with the highest LSN for $k$ in $\mathcal{D}$, and every other key stays empty. In particular, replaying any prefix or suffix of the log after a full pass leaves each projected key at the value of its highest applied LSN.
\end{proposition}

\Cref{app:proof} proves the proposition by induction on the length of $\mathcal{D}$. The invariant is that each key seen so far stores the largest LSN among its delivered events, together with that event's value.

The commit protocol redelivers a batch after any transient sink failure. The delivered set therefore eventually covers every event that maps to a valid key. This holds as long as the worker keeps making progress and Kafka keeps each event until the worker commits it. Each key then holds the value of its highest-LSN event, which is its latest committed row or a soft-delete marker. This result also needs a log that holds every committed change, and that condition fails if the replication slot is dropped or invalidated. A re-snapshot after such a loss cannot repair a key that already holds a streamed LSN, because the worker applies snapshot rows at $\ell=0$ and the guard rejects them. A re-snapshot that carries a different value at LSN zero also breaks assumption (i). The guard then keeps the first snapshot row, and any streamed change still outranks it.

\paragraph{What the guard does not cover.}
The proposition covers one sink and one key at a time, and the paper claims nothing stronger. Redis invalidations can repeat after a redelivery. The Iceberg audit table appends every delivered event, so a redelivered batch appends its rows again. The worker also buffers audit rows until the batch flush. A crash in the middle of a batch can lose those rows, and the redelivery then records events that were already applied as not applied. The worker also skips a content error, which is an event that no key mapping can apply. It audits that event as not applied and commits past it, so redelivery never repairs that key. A reader can also see one key updated before another, because the guard orders changes per key only. The paper therefore describes this path as an idempotent, LSN-guarded projection. It does not claim that each change takes effect once across every destination.

\paragraph{Slot-lag monitoring.}
A logical replication slot makes PostgreSQL retain write-ahead log (WAL) until the consumer confirms it, so a stalled consumer makes the database hold disk. The monitor measures slot lag in bytes with \texttt{pg\_wal\_lsn\_diff}, and a CloudWatch alarm fires at 200~MiB.
\begin{align}
\Delta_{\mathrm{slot}}(t) &= \mathrm{LSN}_{\mathrm{current}}(t)-\mathrm{LSN}_{\mathrm{confirmed}}(t), \label{eq:slotlag}\\
\text{alarm at minute } t &\iff \min_{j\in\{0,1,2\}} M_{t-j}\ge\theta_{\mathrm{slot}}. \label{eq:slotalarm}
\end{align}
Here $M_t$ is the largest slot lag that the monitor publishes in one-minute period $t$, so the alarm needs three consecutive periods at or above the threshold. The threshold $\theta_{\mathrm{slot}}$ equals $200\times 1024\times 1024$ bytes, which is 209,715,200~bytes. The monitor measures lag from the slot's restart position when the slot has no confirmed position. A slot can also fall behind while its captured tables are quiet, because the connector confirms a position only when it receives a change. A Debezium heartbeat query that upserts one row handles that case, and \cref{sec:eval-cdc} reports its effect.

\section{Evidence-Gated Promotion}\label{sec:promotion}

\input{figures/fig_promotion}

A candidate model reaches traffic only through the ladder in \cref{fig:promotion}. The ladder has up to four steps, which are a holdout gate, a reward-hacking probe, a shadow comparison, and a weighted canary. A failure at any step stops the candidate. A candidate cannot skip the holdout gate, the shadow comparison, or the canary. The probe runs only when the caller supplies a probe result, so supervised retraining candidates skip it.

\paragraph{Holdout gate.}
The holdout gate requires an AUC of at least 0.85, a mean CRPS of at most 1.0, and a calibration ratio inside $[0.8,1.2]$. The AUC compares binary anomaly labels with the candidate's anomaly scores. The model also predicts a Gaussian $\mathcal{N}(\mu,\sigma^2)$ for each observation $y$, and CRPS has a closed form for that case \citep{gneiting2005crps}.
\begin{align}
\mathrm{CRPS}\big(\mathcal{N}(\mu,\sigma^2),\,y\big) &= \sigma\left[z\,\big(2\Phi(z)-1\big)+2\varphi(z)-\frac{1}{\sqrt{\pi}}\right], \qquad z=\frac{y-\mu}{\sigma}, \label{eq:crps}\\
\rho &= \frac{\overline{e^2}}{\overline{\sigma^2}}, \qquad e=y-\mu. \label{eq:calib}
\end{align}
Here $\Phi$ and $\varphi$ are the standard normal distribution and density functions, and a bar denotes a mean over the holdout set. The CRPS threshold is in the units of $y$. A ratio $\rho$ above 1 means that the errors are larger than the predicted variance, and a ratio below 1 means that they are smaller. Every threshold must pass, so a strong AUC cannot make up for poor calibration.

\paragraph{Reward-hacking probe.}
The probe compares a preference-tuned candidate with the baseline on a held-out set. Let $\bar R$ be the mean total reward, $u$ the rate of structural-constraint violations, and $\bar s$ the mean structural term. A sample violates the structural constraint when its structural term falls below zero.
\begin{align}
\text{block} &\iff \bar R_{\mathrm{cand}}>\bar R_{\mathrm{base}}\ \land\ \big(u_{\mathrm{cand}}>u_{\mathrm{base}}\ \lor\ \bar s_{\mathrm{cand}}<\bar s_{\mathrm{base}}\big). \label{eq:probe}
\end{align}
The rule blocks a candidate whose reward rises while its physical consistency gets worse. The second clause catches a candidate that keeps the same number of violations and makes each one more severe.

\paragraph{Shadow comparison.}
The shadow step scores the same $N$ inputs with the champion and the candidate. It passes when the mean absolute score difference is at most a limit $\delta_{\mathrm{sh}}$ that the caller sets, and the drills set it to 0.05.
\begin{align}
\frac{1}{N}\sum_{j=1}^{N}\big|s^{\mathrm{champ}}_j-s^{\mathrm{cand}}_j\big| &\le \delta_{\mathrm{sh}}. \label{eq:shadow}
\end{align}

\paragraph{Canary and burn rate.}
The canary moves the candidate through the weights 5, 25, 50, and 100. At each weight, a burn-rate check reads the score-success service level objective (SLO). Consider a window $W$ with $\mathrm{bad}_W$ failures among $\mathrm{total}_W$ requests and an SLO target $T$. The burn rate divides the observed bad ratio by the allowed error fraction (1 minus the SLO target) \citep{thurgood2018alerting}.
\begin{align}
B_W &= \frac{\mathrm{bad}_W/\mathrm{total}_W}{1-T}, \qquad \text{tier fires} \iff B_{W_{\mathrm{long}}}>\theta\ \land\ B_{W_{\mathrm{short}}}>\theta. \label{eq:burn}
\end{align}
A window with no requests has $B_W=0$. The system uses three tiers. A burn rate above 14.4 over 1~h and 5~min pages, and a burn rate above 6 over 6~h and 30~min also pages. A burn rate above 3 over 1~day and 2~h raises a warning. The two paging tiers match the recommended tiers of \citet{thurgood2018alerting}, and the warning tier is the project's own choice. A page at any canary weight reverts the full weight to the champion in one action, and a warning does not revert.

\paragraph{Drift and the feedback loop.}
The code includes three detectors that watch for drift after promotion. They cover input shift, calibration, and concept drift, and the concept-drift detector works only through a proxy signal. The input check compares a reference sample with a current sample for each numeric feature, and \cref{app:drift} gives its test and thresholds.

The feedback loop ends at preference data. A builder turns analyst labels of \texttt{incorrect} into offline preference records, and no code path retrains a model from them. No reinforcement-learning fine-tune ran, so no real candidate has reached the probe. The live endpoint allows one variant, so the canary is traffic-inert there. The shadow and canary results in \cref{sec:eval-gates} therefore come from synthetic drills.

\section{Scale-to-Zero Serving}\label{sec:scale}

In the design, the heavier model runs behind a SageMaker asynchronous inference endpoint. The endpoint holds a request queue, so it can accept a request while it has no instance and serve that request after it wakes. Its autoscaling target allows between zero and one instance, and two policies share that target. Both policies use \texttt{EndpointName} as their only metric dimension.
\begin{align}
\frac{\mathrm{backlog}(t)}{n_{\mathrm{inst}}(t)} &\rightarrow 5 \quad (n_{\mathrm{inst}}\ge 1), \qquad n_{\mathrm{inst}} \gets 1 \iff \min\big(\bar H_{k-1},\,\bar H_k\big)\ge 1 \quad (n_{\mathrm{inst}}=0). \label{eq:scale}
\end{align}
\begin{itemize}[leftmargin=*,itemsep=2pt,topsep=2pt]
\item \textbf{Track.} A target-tracking policy steers the metric \texttt{Approximate\allowbreak Backlog\allowbreak Size\allowbreak Per\allowbreak Instance} toward a setpoint of 5. When the backlog empties, it scales the endpoint in to zero instances.
\item \textbf{Wake.} Target tracking cannot see the move from zero instances, so a separate step policy handles it. The step policy sets the capacity to exactly 1 when the average $\bar H_k$ of \texttt{Has\allowbreak Backlog\allowbreak Without\allowbreak Capacity} is at least 1 in two consecutive 60~s periods, and its cooldown is 60~s.
\end{itemize}
The target caps capacity at one instance, so the tracking policy in practice only scales the endpoint in to zero.

\paragraph{Two dimension bugs.}
Both policies had a bug in their CloudWatch metric dimensions. An AI audit session compared the SageMaker Terraform with AWS's own canonical examples and returned 19 numbered findings. The session recommended changes and edited nothing itself. Its first finding, which the public repository's incident log records as P25, was that the target-tracking metric specification lacked the \texttt{EndpointName} dimension. SageMaker publishes the metric under that dimension, so the policy's alarms queried a series that does not exist. Those alarms never received data, and scale-in to zero could never fire. The audit found the gap from the code alone, before any live symptom appeared, and the fix added the dimension.

In a bounded AWS window, a live check of that fix exposed a second bug, P26. The wake alarm carried an extra \texttt{VariantName} dimension, and SageMaker does not publish \texttt{Has\allowbreak Backlog\allowbreak Without\allowbreak Capacity} under it. The alarm received no datapoints for more than 15~minutes while a real backlog waited, so the endpoint would never have woken from zero. P26 is not one of the 19 audit findings, and its fix removed the extra dimension. Both bugs left every resource present and every apply clean. Neither bug raised an error, and only a live check of the scaling behavior exposed P26.

\paragraph{Evidence on the endpoint.}
In a bounded AWS window, a controlled drill woke the endpoint from zero instances. The backlog alarm fired, capacity moved to one, and the same request returned a score. The endpoint then returned to zero instances. The record gives no wake time, so this paper reports none. The endpoint served a labeled demo stand-in image, and AWS rejected a second variant on the asynchronous endpoint. The endpoint therefore stays single-variant.

\paragraph{The Kubernetes path.}
A KEDA ScaledObject on EKS gives a second scale-from-zero path for a serving deployment. This path ran once in a bounded AWS window and was then torn down, and ECS stayed the front door. \Cref{app:k8s} reports that test.

\section{Evaluation}\label{sec:eval}

This section reports what was measured and where it ran. \Cref{tab:results} lists each result with its scope label and the kind of artifact that holds it. BOUNDED-AWS-WINDOW marks a measurement on real AWS resources during one attended, time-boxed window, and its evidence is private. LOCAL marks a run on a local machine or a local stack. SYNTHETIC-DRILL marks a deterministic drill that runs real project code on synthetic inputs without any infrastructure. PUBLIC-REPO marks a result that the public repository states or reproduces. No experiment in this section measures detection accuracy. \Cref{app:misc} reports the scoring-kernel benchmark, the audit tally, and the test suites from \cref{tab:results}.

\begin{table}[t]
\caption{Each measured result appears with its scope. PUBLIC-REPO means that the public repository states or reproduces the result. A reader can rerun the public test suite, the local CDC tests, and the public scripts for the skew, canary, and slot drills. The soak row carries PUBLIC-REPO because the public README restates it, and it cannot be rerun. The AWS rows come from bounded windows whose evidence is private. The row groups follow \cref{sec:scale}, \cref{sec:eval-stream,sec:eval-soak,sec:eval-cdc,sec:eval-gates}, and \cref{app:misc}.}
\label{tab:results}
\centering
\footnotesize
\setlength{\tabcolsep}{3pt}
\begin{tabular}{@{}>{\raggedright\arraybackslash}p{3.3cm}>{\raggedright\arraybackslash}p{5.25cm}>{\raggedright\arraybackslash}p{2.8cm}>{\raggedright\arraybackslash}p{1.95cm}@{}}
\toprule
Result & Value & Scope & Artifact \\
\midrule
SageMaker wake from zero & 0 to 1 to 0 instances, no wake time recorded & \scope{BOUNDED-AWS-WINDOW} & acceptance record \\
\addlinespace[2pt]
Stream run load & 34,123 records in 900~s, 5~records/s with a 400~records/s burst from 60 to 135~s & \scope{BOUNDED-AWS-WINDOW} & load capture \\
Backpressure, source and key-by vertex & 0.927 in the burst window, 0.998 after it, below 0.1 at about 694~s & \scope{BOUNDED-AWS-WINDOW} & Flink capture \\
Lag, one-minute maxima & 139~s for Kinesis and for Flink, then 0 & \scope{BOUNDED-AWS-WINDOW} & CloudWatch export \\
Scale request to first HTTP 200, KEDA on EKS & 1.89~s & \scope{BOUNDED-AWS-WINDOW} & scale timeline \\
\addlinespace[2pt]
One-hour soak outcomes & 35,999 of 36,000 HTTP 200, one 503 & \scope{BOUNDED-AWS-WINDOW}, \scope{PUBLIC-REPO} & run summary, README \\
Soak latency, p95 and p99 & client 142.751 and 240.283~ms, server 7.048 and 9.348~ms & \scope{BOUNDED-AWS-WINDOW} & run summary \\
\addlinespace[2pt]
CDC acceptance on MSK Serverless & drills A to E passed, replay hash unchanged & \scope{BOUNDED-AWS-WINDOW} & acceptance record \\
Local CDC end-to-end tests & 5 passed in 36.41~s on a fresh stack & \scope{LOCAL}, \scope{PUBLIC-REPO} & drill record, README \\
Insert to read store & 0.57~s against a 5~s target & \scope{LOCAL} & drill record \\
Slot growth and drain & 74,035,000~bytes, then 0 & \scope{LOCAL} & drill record \\
Stale event after a newer one & guarded severity held at 0.95 & \scope{SYNTHETIC-DRILL} & drill record \\
\addlinespace[2pt]
Candidate with serving skew & holdout passed, shadow 0.6056 against 0.05 & \scope{SYNTHETIC-DRILL} & drill record \\
Canary with planted regression & reverted at weight 25 & \scope{SYNTHETIC-DRILL} & drill record \\
\addlinespace[2pt]
Scoring kernel latency & p95 0.6109~ms on one thread & \scope{LOCAL} & benchmark record \\
AI audit of SageMaker Terraform & 19 findings in four groups & \scope{LOCAL} & audit record \\
Public snapshot test suite & 821 passed, 12 skipped, 83.72\% coverage & \scope{PUBLIC-REPO} & README \\
Private regression suite & 1,633 passed, 21 skipped & \scope{LOCAL} & gate record \\
\bottomrule
\end{tabular}
\end{table}

\subsection{Stream run under load}\label{sec:eval-stream}

The stream run tested the pipeline under a burst in a bounded AWS window (\cref{fig:streamrun}). A producer sent 34,123 records into provisioned Kinesis over a 900~s run. The load held 5~records/s, except for a burst of 400~records/s from 60~s to 135~s. An observer polled the Flink backpressure ratio every 2~s and kept 270 samples. Each sample waited for the Flink dashboard to answer, so the samples arrived less often than the poll interval.

\input{figures/fig_stream_run}

Backpressure on the source and key-by vertex first rose above 0.1 at 79.5~s, inside the burst window. Inside that window, its maximum was 0.927 over 22 samples. The ratio kept rising after the burst ended, and it first reached 0.998 about 160~s after the load started. That value is the maximum for the whole run. The ratio cleared below 0.1 at about 694~s, and the final three samples read 0.0. The downstream vertex reported 0.0 throughout. That vertex chains the keyed operator with a print sink, and the keyed operator runs feature computation, the physics gate, a DynamoDB write, and a signed scorer call inline.

CloudWatch reports lag in one-minute buckets, so the lag figures here are one-minute bucket maxima. Kinesis iterator age peaked at 139,000~ms in a single bucket, and that bucket's average was 2,278.7~ms, so the peak was brief. Flink's \texttt{millisBehindLatest} also peaked at 139,000~ms, and it held that maximum in four consecutive buckets. Both series then drained to 0. In the same window, KEDA on EKS took 1.89~s from the scale request to the first HTTP 200, and \cref{app:k8s} explains why this measures a warm cluster.

The run measured the pipeline under load, and it did not measure detection quality. The evidence for this run does not record whether the inline scorer calls succeeded, so the paper makes no claim about them.

\subsection{One-hour soak through the ECS front door}\label{sec:eval-soak}

The soak tested serving capacity in a bounded AWS window. After a warmup of 60~s and 600 requests, one 3,600~s trial offered 10 requests per second with open-loop arrivals and no retries. Each request passed through the IAM-signed API Gateway route and the VPC Link to the ECS service. The trial returned 35,999 HTTP 200 responses and one HTTP 503, for a success ratio of 0.99997 against a 0.999 gate. Goodput counts a request only when it succeeds within 1,000~ms of its scheduled start, and 35,996 requests met that bound.

Over the 35,999 successes, client latency had a p50 of 101.418~ms, a p95 of 142.751~ms, and a p99 of 240.283~ms, with a maximum of 1,474.254~ms. Server latency had a p50 of 5.219~ms, a p95 of 7.048~ms, and a p99 of 9.348~ms, with a maximum of 45.451~ms. The client called the public API Gateway endpoint, so client latency includes the network path, and server latency is the time inside the service. The largest start lag of the client was 60.583~ms, below its 100~ms limit, so the load generator kept its schedule.

The soak repeated one fixed request whose history held one fix. The planner therefore ran the corridor, prism, speed, and validate steps, and it never ran the gap detector or the $\Pdata$ term. The soak also started prewarmed after a successful preflight, so it supports no cold-start claim. It measures serving capacity, and it measures no detection quality.

\subsection{Change data capture}\label{sec:eval-cdc}

A bounded AWS window ran the CDC path on an MSK Serverless cluster, with Debezium Connect and the worker on Fargate. Five drills passed. This evidence is private, and the public repository cannot reproduce it. Drill A flagged a controlled watchlist entry, and the scorer returned a score of 1 with human review required. The entry was then deleted. Drill B replayed the topic with a fresh consumer group. The replay read 3 tombstones, which the worker skips, and the LSN guard rejected 3 change events. The 3 DynamoDB items and their content hash did not change. Drill C delivered 4 writes that spanned a forced redeployment of Kafka Connect. Drill E paused Connect, and 48 writes built a slot lag of 27,408 to 27,464~bytes. The alarm fired after 3 one-minute datapoints on a temporary 16~KiB threshold, and the normal threshold of 209,715,200~bytes was then restored. The drill tried other thresholds only to calibrate the alarm. The drill D cleanup later soft-deleted the 52 records that drills C and E wrote. After the drills, Terraform reconciliation reported no changes, and the connector and its task were running.

The same window exposed a slot that grew while the connector still reported RUNNING. Before the repair, slot lag was 1,006,637,152~bytes. A Debezium heartbeat query now upserts one row, so the slot advances when the source tables are quiet. After the fix, slot lag fell to 4,328~bytes. The same path also runs on a local stack, and \cref{app:local-cdc} reports those runs.

Two drills isolate the failures that the design guards against. In a local slot drill on a throwaway PostgreSQL 16 container, a logical slot with no consumer grew at every one of five write rounds, from 6,926,136 to 74,035,000~bytes. The lag alert fired at a drill threshold of 65,536~bytes, and draining the slot returned the lag to 0. A synthetic duplicate drill ran the guard against an in-memory sink. When a stale event at LSN 2000 arrived after LSN 3000, the unguarded sink regressed the severity from 0.95 to 0.9, and the guarded sink kept 0.95. After a crash between the sink acknowledgement and the offset commit, the guard rejected all five redeliveries, and the state hash matched the baseline.

\subsection{Promotion gate drills}\label{sec:eval-gates}

Two synthetic drills run the real gate code with a synthetic toy scorer and synthetic AIS-shaped inputs, with no AWS. In the skew drill, a candidate had a bug in its online feature encoding. The holdout gate passed it with an AUC of 1.0000, a CRPS of 0.0563, and a calibration ratio of 0.9955. The drill derives its labels by thresholding the same scores, and it draws its observations from the predicted Gaussian. These holdout values therefore follow from the construction and do not measure a model. The shadow comparison then failed the candidate, with a mean absolute difference of 0.6056 and a maximum of 0.9363 against the 0.05 limit. The holdout set could not see a bug that lived only in the serving path, and the shadow step caught it.

In the canary drill, a candidate passed the holdout gate with an AUC of 1.0000, a CRPS of 0.5807, and a calibration ratio of 1.0567. This drill also builds its holdout set so that the candidate passes. The candidate then passed the shadow step with a mean absolute difference of 0.0088. The canary advanced to weight 5 and met a planted regression at weight 25. The rollback step then returned the full weight to the champion, and the final status was \texttt{rolled\_back}. The drill injects a burn check that reports a page at weight 25, so it tests the revert path and does not exercise \cref{eq:burn}.

\section{Lessons and Limitations}\label{sec:lessons}

\paragraph{What is not built.}
The paper claims none of the following as running.
\begin{itemize}[leftmargin=*,itemsep=1pt,topsep=2pt]
\item Candidate V3, a model candidate, is not deployed. The AWS endpoint served only a labeled demo stand-in image, so no trained model produced a served score.
\item No reinforcement-learning fine-tune ran, and the feedback loop ends at offline preference data.
\item Distributed tracing was never built, and the paper reports no drift monitoring on AWS.
\item CDC on AWS ran only in bounded acceptance windows. The canary is single-variant and traffic-inert on the live endpoint.
\item The stream carried synthetic load, and the ingestor replays only a synthetic fixture. The Flink job has no event-time windows or watermarks, and its gate scores the fast jumps in \cref{sec:scoring} late or drops them.
\item The lane release publisher is assumed integration, and the review console ran only locally.
\item The Workflow~5 job service, dispatch worker, and result collector are not built.
\end{itemize}

\paragraph{Lessons from the evidence.}
A resource that exists and raises no error is weak evidence that its mechanism works. The P25 and P26 alarms were present and quiet while scale-in and wake-up could never fire. The useful check is a live test of the behavior that the resource provides.

An HTTP 200 does not show that the intended detector ran. In one AWS streaming window, Flink kept one fix per vessel, so the planner never added the gap detector while every request still returned 200. The fix keeps five fixes, which clears the three-fix threshold of the gaps step. A code review found a related silent fallback. The AWS task definition never injected the database credential, so the review queue had fallen back to process memory.

A replication slot holds disk until a consumer drains it. The local slot drill grew a slot at every write round, and the AWS window found a slot at 1,006,637,152~bytes while the connector reported RUNNING. Slot lag therefore needs its own alarm, and quiet source tables need a heartbeat.

A holdout gate cannot see a bug that lives only in the serving path. The skew drill showed this case directly. A summary can also misstate its artifact, so this paper takes every number from the raw artifact and labels its scope. For example, 0.927 is the burst-window maximum, and the raw capture shows the run maximum of 0.998 after the burst.

\section{Conclusion}\label{sec:conclusion}

Harbormaster keeps each guarantee small and checkable, so that its alerts stay attributable, reviewable, and recoverable. Closed-form rules produce every score, and the paper states which jumps the stream gate scores late or drops. The CDC path gives a per-key replay guarantee with a short proof, and it names the effects that can repeat. The promotion ladder turns model trust into gates that produce evidence. In synthetic drills, the shadow step caught a skewed candidate that the holdout gate passed, and the canary reverted a planted regression. The measured results are small and scoped, and each one names where it ran. The next steps are a trained model on the endpoint, a canary that can move traffic, and CDC that runs outside attended windows.

\bibliographystyle{plainnat}
{\small \bibliography{refs}}

\appendix
\crefalias{section}{appendix}\crefalias{subsection}{appendix}

\section{Reproducibility and Disclosure Boundary}\label{app:repro}

\paragraph{Code and presentation materials.}
The public code is available at \url{https://github.com/arunshar/harbormaster-maritime-ai}. The narrated demonstration is available at \url{https://www.youtube.com/watch?v=NkcvSzdzq8E}. The presentation slides are available at \url{https://claude.ai/artifact/8GTHAQe9Uuyu6bfMkyorD3}. The video and slides explain the system, and the reproducible components are specified below.

\paragraph{What the public repository reproduces.}
The public repository is \url{https://github.com/arunshar/harbormaster-maritime-ai}, and it carries the MIT License. It holds the code, the tests, the Terraform modules, the synthetic replay fixture, the load-test scripts, and the scripts for the skew, canary, and slot drills. A reader can run its unit and integration suite without cloud credentials. At public commit fcac977, that suite reported 821 passed and 12 skipped tests, and later commits to the repository may change the count. A reader with Docker, kind, and kubectl can start the local CDC stack and run its five end-to-end checks. The coverage figure of 83.72\% comes from a local run, because CI does not run coverage. The kernel benchmark script is public, and its measured numbers are not.

\paragraph{What stays private.}
The AWS evidence stays private. It includes the stream-run captures, the soak summary, the CDC acceptance record, the SageMaker acceptance record, and the teardown plan. These files contain account and resource identifiers, so the paper cites them by kind only. The public README states the soak outcome and the client latency. It also states that the repository claims no end-to-end CDC run on AWS. The AI audit record, the private regression suite, the drill transcripts, the duplicate-drill script, and the local CDC run records are private as well.

\paragraph{What the paper never claims.}
\begin{itemize}[leftmargin=*,itemsep=1pt,topsep=2pt]
\item The paper claims no detection accuracy, no false-alert rate, and no result on real labeled cases.
\item It never claims that a trained model served traffic on AWS.
\item It claims no input other than synthetic load and the synthetic replay fixture.
\item It never claims that each change takes effect once across every destination.
\item It claims no SLO, no cold-start time, and no SageMaker wake time.
\item It never claims that CDC ran on AWS outside bounded windows, and it reports no drift monitoring on AWS.
\end{itemize}

\paragraph{Build record.}
One git author made the 226 commits behind this paper. GitHub lists 49 pull requests for the private repository, and 48 were merged. Coding agents served as tooling. Git history holds 44 merge commits, and 38 of them carry the agent's branch prefix. Agents ran no Terraform apply and no push. The author ran every apply, push, login, and live window, and this record is self-reported.

\section{The Full System Board}\label{app:overview}

\Cref{fig:overview} places all five workflows on one page, together with the stores that they share. The board shows the integrated target design, and it does not show a deployed system.

\input{figures/fig_overview}

\section{Method Details}\label{app:method}

\Cref{alg:planner} gives the rule that the scoring path in \cref{sec:scoring} follows for one report.

\begin{algorithm}[ht]
\caption{The deterministic scoring path handles one report.}
\label{alg:planner}
\begin{algorithmic}[1]
\Require current fix $x$, history $h$ of prior fixes
\State Sort $(h,x)$ by time, drop fixes whose timestamp repeats, and let $n$ be the number of remaining fixes minus one
\If{any segment of the sorted track needs a speed above $12\,\vmax$}
  \State \Return HTTP 422 with no score
\EndIf
\State $\mathcal{P}\gets\{\text{corridor}\}$
\If{$n\ge 1$}
  \State $\mathcal{P}\gets\mathcal{P}\cup\{\text{prism},\ \text{speed},\ \text{validate}\}$
\EndIf
\If{$n\ge 3$}
  \State $\mathcal{P}\gets\mathcal{P}\cup\{\text{gaps}\}$
\EndIf
\State Run every detector in $\mathcal{P}$ and the watchlist lookup, and collect the severities $s_i$
\State Compute $S$ and $c$ with \cref{eq:fusion}
\State Queue a review case when \cref{eq:review} holds
\State \Return $S$, $c$, the reasons, and a trace identifier
\end{algorithmic}
\end{algorithm}

\subsection{Proof of the replay proposition}\label{app:proof}

\begin{proof}
The proof uses induction on the length of $\mathcal{D}$ with one invariant. After any prefix of $\mathcal{D}$, each key seen so far stores the largest LSN $m(k)$ among its delivered events, together with that event's value, and every unseen key is empty. The invariant holds for the empty prefix. Suppose the next delivery is $(k,\ell,v)$. Assumption (iv) ensures that the delivery reaches the guard for key $k$. If $k$ is empty, the condition holds, the write succeeds, and $m(k)=\ell$. If $\ell^\star(k)<\ell$, the write succeeds, and $\ell$ becomes the new maximum. If $\ell^\star(k)\ge\ell$, the condition fails and the item keeps its state. The maximum also stays the same, because $\ell\le m(k)$. When $\ell=m(k)$, assumption (i) says that the stored value already equals $v$. Assumption (ii) ensures that no other write changes $\ell^\star(k)$ between the check and the put. Assumption (iii) ensures that the stored value comes from a single event and never mixes fields from two events. Each write touches one item, so every other key keeps the invariant. The final state therefore depends only on the set of delivered events. A replayed prefix or suffix adds no event outside that set, so the state stays the same.
\end{proof}

\subsection{The input drift check}\label{app:drift}

The input check in \cref{sec:promotion} compares a reference sample with a current sample for each numeric feature. It bins both samples into at most 10 equal-frequency bins whose edges are the distinct reference deciles. The outer edges are open, and each fraction has a floor of $10^{-6}$.
\begin{align}
\mathrm{PSI} &= \sum_{b=1}^{B}\big(q^{\mathrm{cur}}_b-q^{\mathrm{ref}}_b\big)\ln\frac{q^{\mathrm{cur}}_b}{q^{\mathrm{ref}}_b}, \quad B\le 10, \qquad D_{\mathrm{KS}}=\sup_x\big|F_{\mathrm{ref}}(x)-F_{\mathrm{cur}}(x)\big|, \label{eq:psi}\\
p &= Q_{\mathrm{KS}}\!\left(\Big(\sqrt{n_e}+0.12+\frac{0.11}{\sqrt{n_e}}\Big)D_{\mathrm{KS}}\right), \qquad n_e=\frac{nm}{n+m}. \label{eq:ks}
\end{align}
Here $q^{\mathrm{ref}}_b$ and $q^{\mathrm{cur}}_b$ are the reference and current fractions in bin $b$, and $F$ is an empirical distribution function. The function $Q_{\mathrm{KS}}$ is the survival function of the Kolmogorov distribution, and $n$ and $m$ are the two sample sizes. This p-value uses a standard asymptotic approximation to the Kolmogorov distribution with the effective sample size $n_e$. A feature drifts when $\mathrm{PSI}\ge 0.25$, or when $\mathrm{PSI}\ge 0.1$ and $p<0.05$.

\section{Additional Measurements}\label{app:measure}

\subsection{Local change data capture runs}\label{app:local-cdc}

On the local stack, PostgreSQL 16 with logical decoding feeds Strimzi-managed Kafka on kind. In the first local run, a new watchlist row reached DynamoDB Local 0.57~s after the PostgreSQL insert, against a 5~s target, and all five end-to-end tests passed in 33.50~s. A second run on a fresh stack followed a change to the registry key schema. Its smoke test took 4.28~s, and the five tests passed in 36.41~s. Two attempts before this fresh-stack run recorded one slow result of 5.06~s and one warm-topic pass of 0.55~s.

\subsection{Kernel latency, audit, and test suites}\label{app:misc}

A local benchmark timed the single-event scoring kernel on one thread over 1,000 iterations after 200 warmup iterations. The request carried a 121-fix history, the watchlist was off, and no endpoint supplied $\Pdata$. The kernel had a p95 of 0.6109~ms and a p99 of 0.7483~ms, and its throughput was 1,760.2~scores/s. The public benchmark document describes the method, and it prints no measured numbers.

The AI audit session in \cref{sec:scale} returned 19 findings in four groups. Two findings were marked as required fixes, seven were hardening items, six became decision records, and four needed no change. The second required fix was a SageMaker Model Package Group that nothing created, and a commit provisioned it.

The public snapshot's test command reported 821 passed and 12 skipped tests, with no failures. The 12 skipped tests need a live PostgreSQL database. Local line and branch coverage over the five packages was 83.72\%, above an 80\% floor, and CI does not run coverage. These are the only test figures that a reader can reproduce. The private regression gate passed 1,633 tests and skipped 21. A separate isolated suite of 57 first-party Lambda tests also passed, and the paper never adds the two counts.

\subsection{The Kubernetes scaling test}\label{app:k8s}

In the design, a KEDA ScaledObject can also scale a serving deployment on EKS from a Kinesis trigger, with a minimum of zero replicas. This Kubernetes path was proven once in a bounded AWS window and then torn down. KEDA on EKS moved the deployment from zero replicas to one ready replica, and the time from the scale request to the first HTTP 200 was 1.89~s. That figure measures a warm cluster that scaled one pod, so it is not an SLO and not a cold-start result. The deployment later returned to zero replicas. A scheduled teardown guard then deleted the EKS node group and the cluster, and the window rolled back to ECS. The final Terraform plan showed 0 resources to add, 0 to change, and 0 to destroy. A decision record kept ECS as the front door, because ECS was already the rollback path and needs no cluster operations.

\end{document}

%% file: figures/fig_architecture.tex
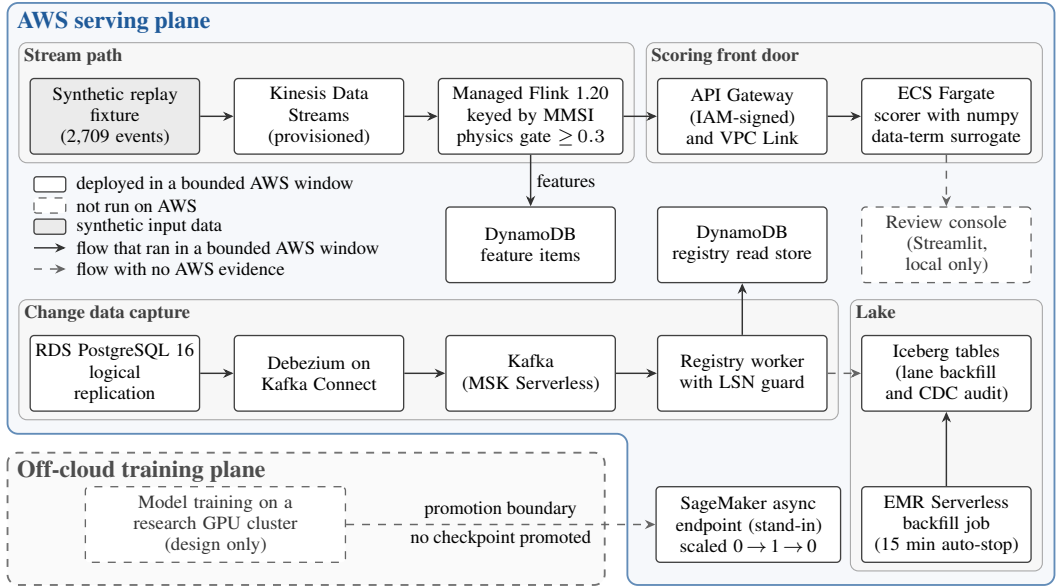
\begin{figure}[t]
\centering
\definecolor{hmaccent}{RGB}{31,90,153}
\definecolor{hmwarm}{RGB}{186,84,20}
\begin{tikzpicture}[
  font=\scriptsize,
  comp/.style={draw=black!75, fill=white, rounded corners=1.5pt, align=center,
    text width=2.1cm, minimum height=1.0cm, inner sep=2pt, line width=0.5pt},
  design/.style={comp, dashed, draw=black!60, text=black!80},
  input/.style={comp, fill=black!8},
  flow/.style={-{Stealth[length=4pt,width=3.5pt]}, line width=0.55pt, draw=black!80},
  designflow/.style={flow, dashed, draw=black!60},
  group/.style={draw=black!35, rounded corners=3pt, fill=black!3, inner xsep=4pt, inner ysep=3pt},
  gtitle/.style={font=\scriptsize\bfseries, text=black!75, anchor=north west, inner sep=2pt},
]
\def\xA{1.42}\def\xB{4.12}\def\xC{6.92}\def\xD{9.72}\def\xE{12.42}
\def\yA{0}\def\yB{-1.7}\def\yC{-3.4}\def\yD{-5.4}

\node[input] (fix) at (\xA,\yA) {Synthetic replay\\fixture\\(2,709 events)};
\node[comp]  (kin) at (\xB,\yA) {Kinesis Data\\Streams\\(provisioned)};
\node[comp, text width=2.3cm]  (fl)  at (\xC,\yA) {Managed Flink 1.20\\keyed by MMSI\\physics gate ${\ge}\,0.3$};
\node[comp]  (gw)  at (\xD,\yA) {API Gateway\\(IAM-signed)\\and VPC Link};
\node[comp]  (ecs) at (\xE,\yA) {ECS Fargate\\scorer with numpy\\data-term surrogate};

\node[comp]   (feat) at (\xC,\yB) {DynamoDB\\feature items};
\node[comp]   (ddb) at (\xD,\yB) {DynamoDB\\registry read store};
\node[design] (con) at (\xE,\yB) {Review console\\(Streamlit,\\local only)};

\node[comp] (pg)   at (\xA,\yC) {RDS PostgreSQL 16\\logical\\replication};
\node[comp] (dbz)  at (\xB,\yC) {Debezium on\\Kafka Connect};
\node[comp] (kfk)  at (\xC,\yC) {Kafka\\(MSK Serverless)};
\node[comp] (sink) at (\xD,\yC) {Registry worker\\with LSN guard};
\node[comp] (ice)  at (\xE,\yC) {Iceberg tables\\(lane backfill\\and CDC audit)};

\node[comp, text width=2.3cm] (sm) at (9.8,\yD) {SageMaker async\\endpoint (stand-in)\\scaled 0\,$\to$\,1\,$\to$\,0};
\node[comp] (emr) at (\xE,\yD) {EMR Serverless\\backfill job\\(15~min auto-stop)};

\node[design, text width=3.3cm] (train) at (2.75,\yD) {Model training on a\\research GPU cluster\\(design only)};

\begin{scope}[on background layer]
  \coordinate (awsTL) at (0,1.5);
  \coordinate (awsTR) at (13.85,1.5);
  \coordinate (awsBR) at (13.85,-6.2);
  \coordinate (awsBM) at (8.2,-6.2);
  \coordinate (awsMM) at (8.2,-4.2);
  \coordinate (awsBL) at (0,-4.2);
  \draw[draw=hmaccent!70, line width=0.7pt, rounded corners=4pt, fill=hmaccent!4]
    (awsTL) -- (awsTR) -- (awsBR) -- (awsBM) -- (awsMM) -- (awsBL) -- cycle;
  \node[group, fit={(fix)(kin)(fl)($(fl.north)+(0,0.36)$)}] (gStream) {};
  \node[group, fit={(gw)(ecs)($(ecs.north)+(0,0.36)$)}] (gScore) {};
  \node[group, fit={(pg)(dbz)(kfk)(sink)($(sink.north)+(0,0.36)$)}] (gCdc) {};
  \node[group, fit={(ice)(emr)($(ice.north)+(0,0.36)$)}] (gLake) {};
  \draw[draw=black!55, dashed, line width=0.7pt, rounded corners=4pt, fill=black!2]
    (0,-4.45) rectangle (7.9,-6.2);
\end{scope}
\node[font=\footnotesize\bfseries, text=hmaccent, anchor=north west, inner sep=2pt] at (0.05,1.47) {AWS serving plane};
\node[font=\footnotesize\bfseries, text=black!70, anchor=north west, inner sep=2pt] at (0.05,-4.48) {Off-cloud training plane};
\node[gtitle] at (gStream.north west) {Stream path};
\node[gtitle] at (gScore.north west) {Scoring front door};
\node[gtitle] at (gCdc.north west) {Change data capture};
\node[gtitle] at (gLake.north west) {Lake};

\draw[flow] (fix) -- (kin);
\draw[flow] (kin) -- (fl);
\draw[flow] (fl) -- (gw);
\draw[flow] (gw) -- (ecs);
\draw[flow] (fl) -- node[right, inner sep=2pt] {features} (feat);
\draw[designflow] (ecs) -- (con);
\draw[flow] (pg) -- (dbz);
\draw[flow] (dbz) -- (kfk);
\draw[flow] (kfk) -- (sink);
\draw[flow] (sink) -- (ddb);
\draw[designflow] (sink) -- (ice);
\draw[flow] (emr) -- (ice);
\draw[designflow] (train.east) -- node[above, inner sep=2pt] {promotion boundary}
  node[below, inner sep=2pt] {no checkpoint promoted} (sm.west);

\begin{scope}[shift={(0.35,-0.9)}]
  \draw[black!75, fill=white, rounded corners=1pt, line width=0.5pt] (0,-0.1) rectangle (0.42,0.1);
  \node[anchor=west, inner sep=1pt] at (0.52,0) {deployed in a bounded AWS window};
  \draw[black!60, dashed, fill=white, rounded corners=1pt, line width=0.5pt] (0,-0.38) rectangle (0.42,-0.18);
  \node[anchor=west, inner sep=1pt] at (0.52,-0.28) {not run on AWS};
  \draw[black!75, fill=black!8, rounded corners=1pt, line width=0.5pt] (0,-0.66) rectangle (0.42,-0.46);
  \node[anchor=west, inner sep=1pt] at (0.52,-0.56) {synthetic input data};
  \draw[flow] (0,-0.84) -- (0.42,-0.84);
  \node[anchor=west, inner sep=1pt] at (0.52,-0.84) {flow that ran in a bounded AWS window};
  \draw[designflow] (0,-1.12) -- (0.42,-1.12);
  \node[anchor=west, inner sep=1pt] at (0.52,-1.12) {flow with no AWS evidence};
\end{scope}
\end{tikzpicture}
\caption{Harbormaster separates an off-cloud training plane from an AWS serving plane. Training on a research GPU cluster is design only, and no trained checkpoint crossed the promotion boundary. The Streamlit review console runs locally. Solid boxes are components that were deployed in a bounded AWS window, and solid arrows are flows that ran in a bounded AWS window. Dashed boxes did not run on AWS, and dashed arrows mark flows with no AWS evidence, such as the append to the CDC audit table. The Terraform at the reviewed commit runs the scorer on Fargate Spot capacity. The SageMaker asynchronous endpoint served a labeled stand-in image and scaled from zero to one instance and back. The ECS task sets no endpoint, so the scorer took its data term $P_{\mathrm{data}}$ from a numpy surrogate on AWS. The figure therefore draws no arrow between the scorer and the endpoint.}
\label{fig:architecture}
\end{figure}

%% file: figures/fig_prism.tex
\begin{figure}[t]
\centering
\definecolor{hmaccent}{RGB}{31,90,153}
\definecolor{hmwarm}{RGB}{186,84,20}
\begin{subfigure}[b]{0.49\linewidth}
\centering
\begin{tikzpicture}[
  font=\footnotesize,
  dim/.style={{Stealth[length=3pt,width=3pt]}-{Stealth[length=3pt,width=3pt]}, line width=0.45pt, draw=black!80},
  ext/.style={black!45, line width=0.3pt},
  fixpt/.style={circle, fill=black, inner sep=0pt, minimum size=3.6pt},
]
\def\aa{1.8}\def\bb{1.3416}\def\cc{1.2}
\fill[hmaccent!10] (0,0) ellipse [x radius=\aa, y radius=\bb];
\draw[hmaccent, line width=0.8pt] (0,0) ellipse [x radius=\aa, y radius=\bb];
\coordinate (P) at ({\aa*cos(118)},{\bb*sin(118)});
\draw[black!55, line width=0.45pt, densely dashed] (-\cc,0) -- (P) -- (\cc,0);
\fill[black!70] (P) circle (1.3pt);
\node[anchor=south east, inner sep=1.5pt] at (P) {$P$};
\draw[hmaccent, line width=0.6pt] (0,0) -- (0,\bb);
\node[anchor=west, inner sep=1.5pt] at (0,0.85*\bb) {$b$};
\draw[hmaccent, line width=0.6pt] (0,0) -- (-\aa,0);
\node[anchor=south, inner sep=1.5pt] at (-0.55,0) {$a$};
\node[fixpt, label={[inner sep=1pt]below:$A$}] at (-\cc,0) {};
\node[fixpt, label={[inner sep=1pt]below:$B$}] at (\cc,0) {};
\draw[ext] (-\aa,0.1) -- (-\aa,1.75);
\draw[ext] (\aa,0.1) -- (\aa,1.75);
\draw[dim] (-\aa,1.6) -- node[above, inner sep=2pt] {$L = v_{\max}\,(t_B - t_A)$} (\aa,1.6);
\draw[ext] (-\cc,-0.45) -- (-\cc,-1.75);
\draw[ext] (\cc,-0.45) -- (\cc,-1.75);
\draw[dim] (-\cc,-1.6) -- node[below, inner sep=2pt] {$d_{AB}$} (\cc,-1.6);
\node[anchor=west, align=left, inner sep=0pt] at (2.05,0.35) {$a = L/2$};
\node[anchor=west, align=left, inner sep=0pt] at (2.05,-0.35) {$b = \tfrac{1}{2}\sqrt{L^2 - d_{AB}^2}$};
\end{tikzpicture}
\caption{The feasible case has $d_{AB} \le L$.}
\label{fig:prism-feasible}
\end{subfigure}\hfill
\begin{subfigure}[b]{0.49\linewidth}
\centering
\begin{tikzpicture}[
  font=\footnotesize,
  dim/.style={{Stealth[length=3pt,width=3pt]}-{Stealth[length=3pt,width=3pt]}, line width=0.45pt, draw=black!80},
  ext/.style={black!45, line width=0.3pt},
  fixpt/.style={circle, fill=black, inner sep=0pt, minimum size=3.6pt},
]
\def\hd{1.8}\def\rr{1.2}
\draw[hmwarm, line width=0.7pt, dashed, fill=hmwarm!7] (-\hd,0) circle (\rr);
\draw[hmwarm, line width=0.7pt, dashed, fill=hmwarm!7] (\hd,0) circle (\rr);
\node[fixpt, label={[inner sep=1pt]below:$A$}] at (-\hd,0) {};
\node[fixpt, label={[inner sep=1pt]below:$B$}] at (\hd,0) {};
\draw[-{Stealth[length=3pt,width=3pt]}, line width=0.45pt, draw=black!80] (-\hd,0) -- node[above, inner sep=1.5pt] {$L/2$} (-\hd+\rr,0);
\draw[-{Stealth[length=3pt,width=3pt]}, line width=0.45pt, draw=black!80] (\hd,0) -- node[above, inner sep=1.5pt] {$L/2$} (\hd-\rr,0);
\node[font=\footnotesize, align=center, text=black!80] at (0,0.95) {no path\\fits $L$};
\draw[ext] (-\hd,-0.45) -- (-\hd,-1.75);
\draw[ext] (\hd,-0.45) -- (\hd,-1.75);
\draw[dim] (-\hd,-1.6) -- node[below, inner sep=2pt] {$d_{AB} > L$} (\hd,-1.6);
\path (0,2.05) -- (0,-2.2);
\end{tikzpicture}
\caption{The infeasible case has $d_{AB} > L$.}
\label{fig:prism-infeasible}
\end{subfigure}
\caption{The space-time prism tests whether one vessel could have moved between two fixes at the speed cap. Fixes $A$ and $B$ lie a distance $d_{AB}$ apart, and the reach length is $L = v_{\max}(t_B - t_A)$, where the code sets $v_{\max}$ to 25~knots. (a) When $d_{AB} \le L$, the feasible positions fill an ellipse with foci $A$ and $B$, semi-major axis $a = L/2$, and semi-minor axis $b = \frac{1}{2}\sqrt{L^2 - d_{AB}^2}$. A point $P$ on the boundary satisfies $|AP| + |PB| = L$. (b) When $d_{AB} > L$, the reach disks of radius $L/2$ around $A$ and $B$ do not meet at the midpoint time, so no path fits within $L$.}
\label{fig:prism}
\end{figure}
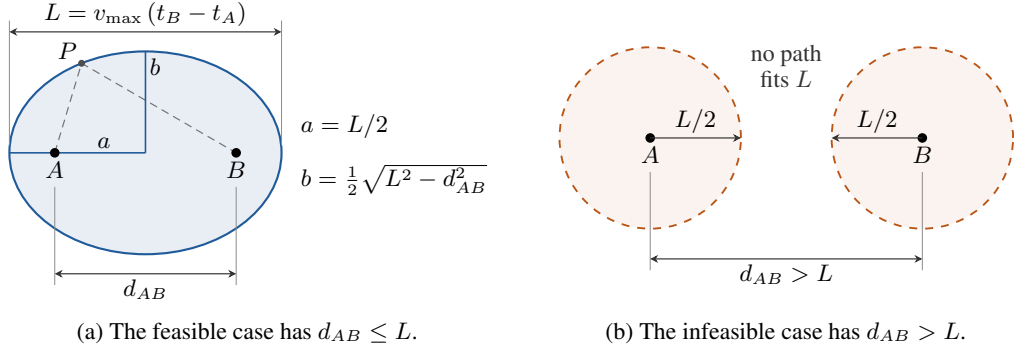

%% file: figures/fig_cdc.tex
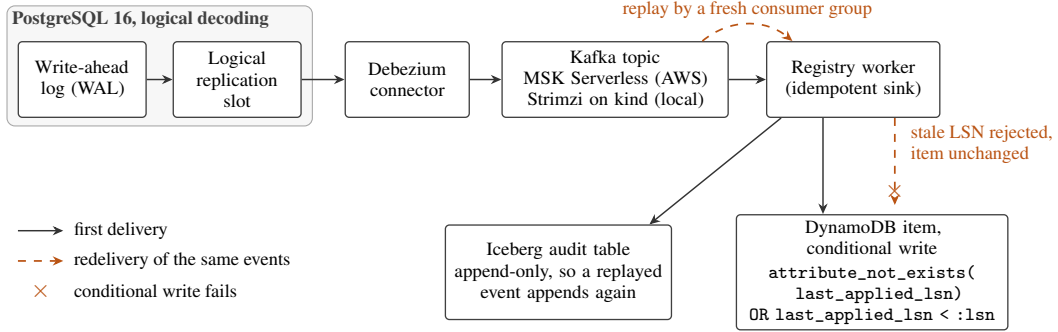
\begin{figure}[t]
\centering
\definecolor{hmaccent}{RGB}{31,90,153}
\definecolor{hmwarm}{RGB}{186,84,20}
\begin{tikzpicture}[
  font=\scriptsize,
  comp/.style={draw=black!75, fill=white, rounded corners=1.5pt, align=center,
    minimum height=1.0cm, inner sep=2pt, line width=0.5pt},
  flow/.style={-{Stealth[length=4pt,width=3.5pt]}, line width=0.55pt, draw=black!80},
  replay/.style={-{Stealth[length=4pt,width=3.5pt]}, line width=0.65pt, draw=hmwarm, dashed},
  group/.style={draw=black!35, rounded corners=3pt, fill=black!3, inner xsep=4pt, inner ysep=3pt},
  gtitle/.style={font=\scriptsize\bfseries, text=black!75, anchor=north west, inner sep=2pt},
]
\node[comp, text width=1.55cm] (wal)  at (1.0,0)  {Write-ahead\\log (WAL)};
\node[comp, text width=1.55cm] (slot) at (3.05,0) {Logical\\replication\\slot};
\node[comp, text width=1.5cm]  (dbz)  at (5.3,0)  {Debezium\\connector};
\node[comp, text width=2.85cm] (kfk)  at (8.05,0) {Kafka topic\\MSK Serverless (AWS)\\Strimzi on kind (local)};
\node[comp, text width=2.15cm] (sink) at (11.2,0) {Registry worker\\(idempotent sink)};

\node[comp, text width=2.85cm, minimum height=1.25cm] (ice) at (7.3,-2.55)
  {Iceberg audit table\\append-only, so a replayed\\event appends again};
\node[comp, text width=3.45cm, minimum height=1.25cm] (ddb) at (11.45,-2.55)
  {DynamoDB item, conditional write\\[1pt]
   \ttfamily attribute\_not\_exists(\\
   \ttfamily \ \ last\_applied\_lsn)\\
   \ttfamily OR last\_applied\_lsn < :lsn};

\begin{scope}[on background layer]
  \node[group, fit={(wal)(slot)($(slot.north)+(0,0.36)$)}] (gPg) {};
\end{scope}
\node[gtitle] at (gPg.north west) {PostgreSQL 16, logical decoding};

\draw[flow] (wal) -- (slot);
\draw[flow] (slot) -- (dbz);
\draw[flow] (dbz) -- (kfk);
\draw[flow] (kfk) -- (sink);
\draw[flow] ($(sink.south west)+(0.2,0)$) -- ($(ice.north east)+(-0.3,0)$);
\coordinate (s1) at ($(sink.south)+(-0.4,0)$);
\coordinate (d1) at (s1 |- ddb.north);
\draw[flow] (s1) -- (d1);

\draw[replay] ($(kfk.north east)+(-0.35,0)$) to[out=40, in=140]
  node[above, inner sep=2pt, text=hmwarm] {replay by a fresh consumer group} ($(sink.north west)+(0.35,0)$);
\coordinate (s2) at ($(sink.south)+(0.55,0)$);
\coordinate (d2) at (s2 |- ddb.north);
\draw[replay] (s2) -- ($(d2)+(0,0.12)$);
\node[text=hmwarm, font=\footnotesize\bfseries, inner sep=0pt] at ($(d2)+(0,0.3)$) {$\times$};
\node[anchor=north west, align=left, text=hmwarm, inner sep=1pt] at ($(s2)+(0.17,-0.05)$) {stale LSN rejected,\\item unchanged};

\begin{scope}[shift={(0.15,-2.0)}]
  \draw[flow] (0,0) -- (0.6,0);
  \node[anchor=west, inner sep=1pt] at (0.7,0) {first delivery};
  \draw[replay] (0,-0.4) -- (0.6,-0.4);
  \node[anchor=west, inner sep=1pt] at (0.7,-0.4) {redelivery of the same events};
  \node[text=hmwarm, font=\footnotesize\bfseries, inner sep=0pt] at (0.3,-0.8) {$\times$};
  \node[anchor=west, inner sep=1pt] at (0.7,-0.8) {conditional write fails};
\end{scope}
\end{tikzpicture}
\caption{The change data capture (CDC) path copies PostgreSQL 16 changes into DynamoDB. Debezium reads the write-ahead log through a logical replication slot and publishes change events to Kafka. Kafka ran as Amazon MSK Serverless in the bounded AWS window and as Strimzi on kind locally. The registry worker writes each item with the conditional write shown, so the write succeeds only when the item has no applied LSN or holds an older one. A replayed event with a stale LSN fails that condition and leaves the item unchanged. The Iceberg audit table appends every delivered event, so its rows can repeat after a redelivery. No AWS record shows the audit append, so \cref{fig:architecture} draws that edge dashed.}
\label{fig:cdc}
\end{figure}

%% file: figures/fig_promotion.tex
\begin{figure}[t]
\centering
\definecolor{hmaccent}{RGB}{31,90,153}
\definecolor{hmwarm}{RGB}{186,84,20}
\begin{tikzpicture}[
  font=\scriptsize,
  st/.style={draw=black!75, fill=white, rounded corners=5pt, align=center,
    minimum height=0.9cm, inner sep=2pt, line width=0.5pt},
  cw/.style={draw=black!75, fill=white, circle, minimum size=0.66cm, inner sep=0pt, line width=0.5pt},
  term/.style={draw=black!75, fill=black!8, rounded corners=5pt, align=center,
    minimum height=0.6cm, inner sep=3pt, line width=0.9pt},
  flow/.style={-{Stealth[length=4pt,width=3.5pt]}, line width=0.55pt, draw=black!80},
  rej/.style={-{Stealth[length=4pt,width=3.5pt]}, line width=0.55pt, draw=hmwarm},
  lab/.style={inner sep=1.5pt, font=\scriptsize},
]
\node[st, text width=1.2cm]  (cand) at (0.72,0)  {Candidate\\model};
\node[st, text width=1.15cm] (hold) at (2.7,0) {Holdout\\gate};
\node[st, text width=1.2cm]  (prb)  at (4.65,0)  {Reward-\\hacking\\probe};
\node[st, text width=1.45cm] (shd)  at (6.75,0) {Shadow\\comparison};
\node[cw] (c5)   at (8.65,0)  {5};
\node[cw] (c25)  at (9.57,0) {25};
\node[cw] (c50)  at (10.49,0) {50};
\node[cw] (c100) at (11.41,0) {100};
\node[st, text width=1.3cm, line width=0.9pt, draw=hmaccent, fill=hmaccent!8] (prom) at (13.0,0) {Promoted};

\begin{scope}[on background layer]
  \node[draw=black!35, rounded corners=3pt, fill=black!3, inner xsep=4pt, inner ysep=3pt,
    fit={(c5)(c100)($(c5.north)+(0,0.34)$)($(c5.south)+(0,-0.05)$)}] (gCan) {};
\end{scope}
\node[font=\scriptsize\bfseries, text=black!75, anchor=north west, inner sep=2pt] at (gCan.north west) {canary weight};

\node[term, minimum width=4.45cm] (rej) at (4.72,1.75) {Rejected, never reaches the canary};
\node[term, minimum width=3.6cm]  (rb)  at (10.03,-1.75) {Rolled back, full weight to champion};

\draw[flow] (cand) -- (hold);
\draw[flow] (hold) -- (prb);
\draw[flow] (prb) -- (shd);
\draw[flow] (shd) -- (c5);
\draw[flow] (c5) -- (c25);
\draw[flow] (c25) -- (c50);
\draw[flow] (c50) -- (c100);
\draw[flow] (c100) -- (prom);
\draw[flow] (hold.south) to[out=-50, in=-130] node[lab, below] {no probe result supplied} (shd.south);

\draw[rej] (hold.north) -- node[lab, right, text=hmwarm] {fail} (hold.north |- rej.south);
\draw[rej] (prb.north) -- node[lab, right, text=hmwarm] {blocked} (prb.north |- rej.south);
\draw[rej] (shd.north) -- node[lab, right, text=hmwarm] {fail} (shd.north |- rej.south);
\foreach \n in {c5,c25,c50,c100} {
  \draw[rej] (\n.south) -- (\n.south |- rb.north);
}
\node[lab, anchor=east, text=hmwarm] at ($(c5.south |- rb.north)!0.5!(c5.south)+(-0.08,0)$) {burn};
\end{tikzpicture}
\caption{The promotion ladder moves a candidate model through up to four checks. Black edges move the candidate forward, and orange edges end the promotion. The candidate must pass the holdout gate, then the reward-hacking probe when a probe result is supplied, and then the shadow comparison. Every gate has a reject edge, and a rejected candidate never reaches the canary. The canary advances through weights 5, 25, 50, and 100 while the burn rate stays healthy, and a burn at any step returns the full weight to the champion. The gate results reported in this paper come from synthetic drills. The live asynchronous endpoint has one variant, so the canary moved no traffic on AWS.}
\label{fig:promotion}
\end{figure}
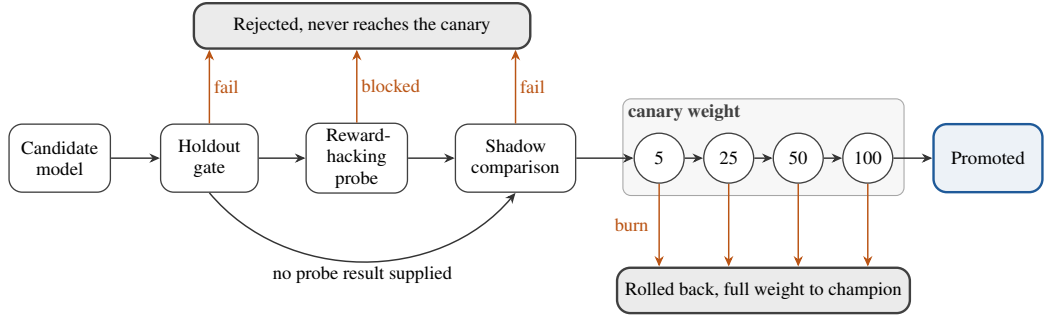

%% file: figures/fig_stream_run.tex
\usepgfplotslibrary{groupplots}
\begin{figure}[t]
\centering
\definecolor{hmaccent}{RGB}{31,90,153}
\definecolor{hmwarm}{RGB}{186,84,20}
\pgfplotstableread[col sep=comma]{figures/data/scale_events.csv}\hmScaleEvents
\pgfplotstablegetelem{0}{t_s}\of\hmScaleEvents\let\hmTScale\pgfplotsretval
\pgfplotstablegetelem{3}{t_s}\of\hmScaleEvents\let\hmTZero\pgfplotsretval
\begin{tikzpicture}[
  burst/.style={fill=black!10, draw=none},
  mScale/.style={hmwarm, line width=0.6pt, dashed},
  mZero/.style={hmwarm, line width=0.6pt, dash dot},
]
\begin{groupplot}[
  group style={group size=1 by 3, vertical sep=5pt,
    xlabels at=edge bottom, xticklabels at=edge bottom},
  scale only axis, width=12.0cm, height=1.8cm,
  xmin=0, xmax=950, xtick={0,100,...,900},
  tick label style={font=\scriptsize}, label style={font=\scriptsize},
  ylabel style={align=center},
  tick align=outside, tick pos=left,
  axis line style={black!70, line width=0.4pt},
  every tick/.style={black!70},
  clip=true,
]
\nextgroupplot[ymin=0, ymax=450, ytick={0,200,400}, ylabel={offered load\\(records/s)},
  legend columns=3,
  legend style={at={(0.5,1.06)}, anchor=south, font=\scriptsize, draw=black!40,
    /tikz/every even column/.append style={column sep=8pt}},
  legend cell align=left]
\addlegendimage{area legend, burst}\addlegendentry{burst window, 60--135 s}
\addlegendimage{black, line width=0.7pt, mark=o, mark size=1.3pt, mark options={fill=white}}\addlegendentry{Kinesis iterator age}
\addlegendimage{hmaccent, line width=0.7pt, mark=square*, mark size=1.1pt}\addlegendentry{Flink \texttt{millisBehindLatest}}
\addlegendimage{mScale}\addlegendentry{KEDA scale request}
\addlegendimage{mZero}\addlegendentry{back to zero replicas}
\fill[burst] ({axis cs:60,0}|-{rel axis cs:0,0}) rectangle ({axis cs:135,0}|-{rel axis cs:0,1});
\addplot[black, line width=0.6pt, const plot, mark=none, forget plot]
  table[col sep=comma, x=t_s, y=records_per_s] {figures/data/stream_rate.csv};
\draw[mScale] ({axis cs:\hmTScale,0}|-{rel axis cs:0,0}) -- ({axis cs:\hmTScale,0}|-{rel axis cs:0,1});
\draw[mZero]  ({axis cs:\hmTZero,0}|-{rel axis cs:0,0}) -- ({axis cs:\hmTZero,0}|-{rel axis cs:0,1});
\node[anchor=north west, font=\scriptsize, inner sep=2pt] at (rel axis cs:0,1) {(a)};
\node[anchor=west, font=\scriptsize, inner sep=2pt] at (axis cs:140,360) {400 records/s};
\node[anchor=south, font=\scriptsize, inner sep=2pt] at (axis cs:620,5) {5 records/s};
\nextgroupplot[ymin=0, ymax=1.08, ytick={0,0.5,1}, yticklabels={0,0.5,1},
  ylabel={backpressure\\(source, key-by)}]
\fill[burst] ({axis cs:60,0}|-{rel axis cs:0,0}) rectangle ({axis cs:135,0}|-{rel axis cs:0,1});
\draw[black!60, dashed, line width=0.5pt] ({rel axis cs:0,0}|-{axis cs:0,0.1}) -- ({rel axis cs:1,0}|-{axis cs:0,0.1});
\addplot[hmaccent, line width=0.6pt, mark=none, forget plot]
  table[col sep=comma, x=t_s, y=backpressure] {figures/data/backpressure.csv};
\draw[mScale] ({axis cs:\hmTScale,0}|-{rel axis cs:0,0}) -- ({axis cs:\hmTScale,0}|-{rel axis cs:0,1});
\draw[mZero]  ({axis cs:\hmTZero,0}|-{rel axis cs:0,0}) -- ({axis cs:\hmTZero,0}|-{rel axis cs:0,1});
\node[anchor=north west, font=\scriptsize, inner sep=2pt] at (rel axis cs:0,1) {(b)};
\node[anchor=south east, font=\scriptsize, inner sep=1.5pt, text=black!70] at (axis cs:905,0.1) {0.1};
\nextgroupplot[ymin=-12, ymax=175, ytick={0,50,100,150}, ylabel={lag (s)\\minute max},
  xlabel={time since load start (s)}]
\fill[burst] ({axis cs:60,0}|-{rel axis cs:0,0}) rectangle ({axis cs:135,0}|-{rel axis cs:0,1});
\addplot[hmaccent, line width=0.7pt, const plot, mark=square*, mark size=1.1pt, forget plot]
  table[col sep=comma, x=t_s, y=flink_millis_behind_s] {figures/data/lag.csv};
\addplot[black, line width=0.7pt, const plot, mark=o, mark size=1.3pt, mark options={fill=white}, forget plot]
  table[col sep=comma, x=t_s, y=kinesis_iterator_age_s] {figures/data/lag.csv};
\addplot[black, line width=0.7pt, mark=none, forget plot] coordinates {(835.482822,0) (895.482822,0)};
\draw[mScale] ({axis cs:\hmTScale,0}|-{rel axis cs:0,0}) -- ({axis cs:\hmTScale,0}|-{rel axis cs:0,1});
\draw[mZero]  ({axis cs:\hmTZero,0}|-{rel axis cs:0,0}) -- ({axis cs:\hmTZero,0}|-{rel axis cs:0,1});
\node[anchor=north west, font=\scriptsize, inner sep=2pt] at (rel axis cs:0,1) {(c)};
\end{groupplot}
\end{tikzpicture}
\caption{The stream run shows offered load, backpressure, and lag from one bounded AWS window. (a) The load generator cycled the synthetic replay fixture into provisioned Kinesis at 5~records/s, with a burst of 400~records/s from 60 to 135~s. Managed Flink consumed the stream. The run measured the pipeline under load and did not measure detection quality. (b) Backpressure on the source and key-by vertex reached 0.927 inside the gray burst window and 0.998 after it. It dropped to the dashed 0.1 line or below near 694~s. (c) CloudWatch reports lag at one-minute resolution, so each step holds a per-minute maximum. Both series peaked at a lag of 139~s and then drained to 0. The vertical lines mark a separate KEDA on EKS scale-from-zero test in the same window. The dashed line marks the scale request at 327.37~s. The first HTTP 200 came 1.89~s later, at 329.26~s, so the figure does not draw it as a separate line. The deployment returned to zero replicas at 918.24~s.}
\label{fig:streamrun}
\end{figure}
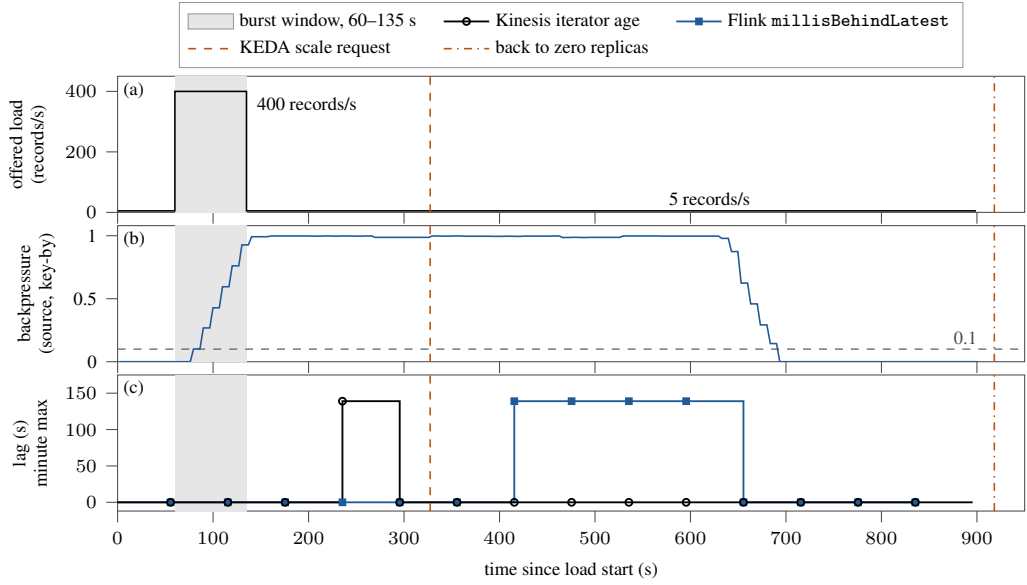

%% file: figures/fig_overview.tex
\begin{figure}[p]
\centering
\includegraphics[width=\textwidth]{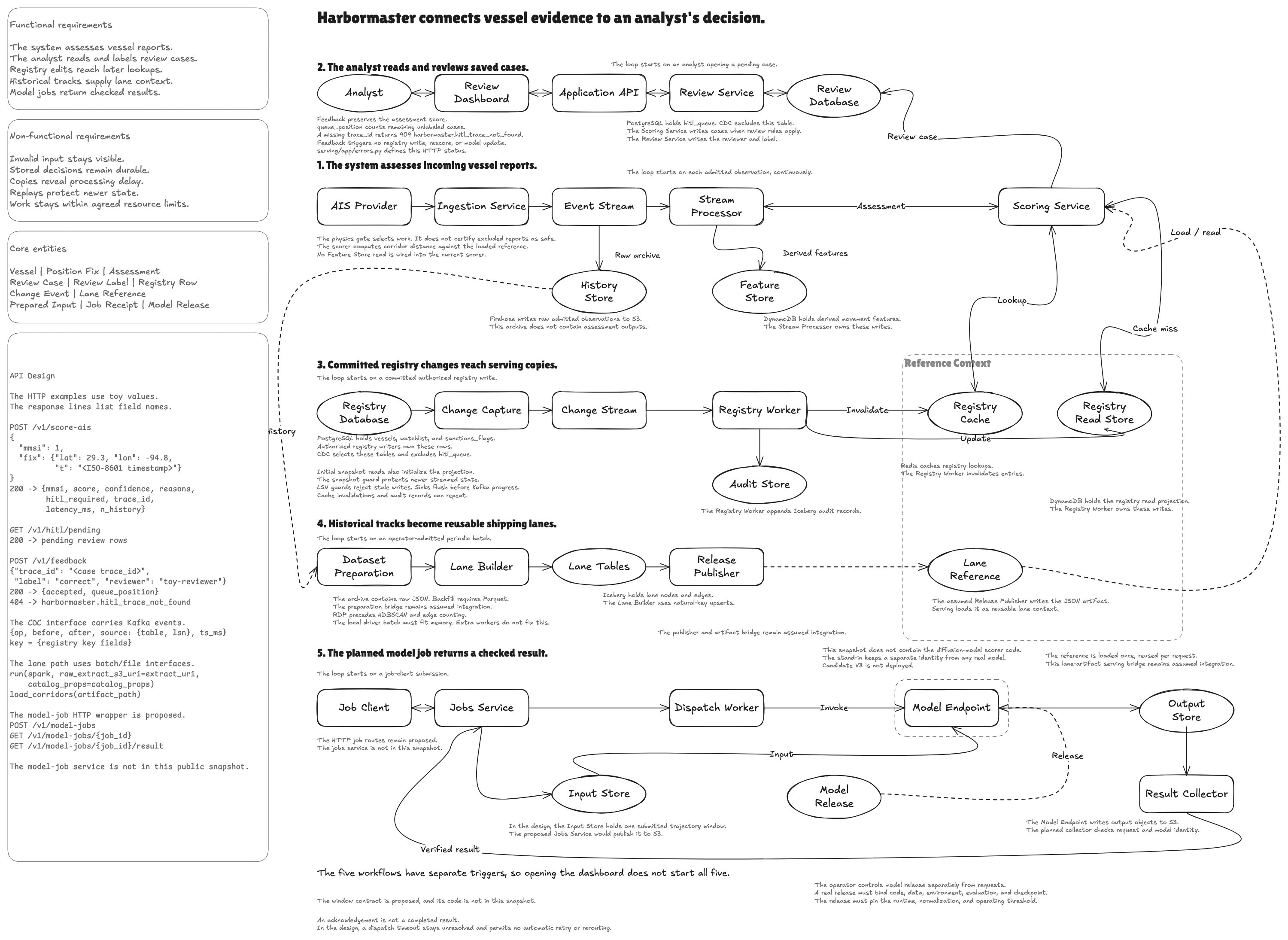}
\caption{The five-workflow board, adapted from the public repository, shows the integrated target design of Harbormaster. It is a design map, not a picture of a deployed system. Several paths on the board are labeled as proposed or as assumed integration, and the board's own notes mark what the public snapshot does not contain. The AIS Provider box is a design placeholder, and only the synthetic replay fixture ran. The board draws the Workflow~5 Jobs Service, Dispatch Worker, and Result Collector with solid lines, and none of them is built. Zoom the PDF to read the notes, and the source SVG is in the public repository.}
\label{fig:overview}
\end{figure}